%% file: AtomicCrosschainTransactions_01.tex
\newcommand{\CLASSINPUTtoptextmargin}{2cm}
\newcommand{\CLASSINPUTbottomtextmargin}{2cm}
\documentclass[journal, a4paper]{IEEEtran}

\newif\ifall
\alltrue 

\usepackage{graphicx}
\usepackage{amsthm}
\usepackage{caption}
\usepackage{listings}

\usepackage{cite}

\newtheorem{theorem}{Theorem}

\begin{document}
\bstctlcite{IEEEexample:BSTcontrol}

\title{Atomic Crosschain Transactions}

\author{
     \IEEEauthorblockN{
    	Peter Robinson\IEEEauthorrefmark{1}\IEEEauthorrefmark{2},
    	Raghavendra Ramesh\IEEEauthorrefmark{1}, 
    	John Brainard\IEEEauthorrefmark{1},
    	Sandra Johnson\IEEEauthorrefmark{1}\IEEEauthorrefmark{3}\\
         } 
    \IEEEauthorblockA{\IEEEauthorrefmark{1}Protocol Engineering Group and Systems (PegaSys), ConsenSys\\}
    \IEEEauthorblockA{\IEEEauthorrefmark{2}School of Information Technology and Electrical Engineering, University of Queensland, Australia\\}
    \IEEEauthorblockA{\IEEEauthorrefmark{3}ACEMS, Queensland University of Technology, Australia\\}
    \IEEEauthorblockA{
    	peter.robinson@consensys.net,
    	raghavendra.ramesh@consensys.net,
    	johngbrainard@gmail.com,
	sandra.johnson@consensys.net\\
	}
    \IEEEauthorblockA{
    	April 2, 2020 (Version 1.0)
	}	
}

\maketitle
\input{abstract}

\IEEEpeerreviewmaketitle

\input{introduction}
\input{components}

\input{processing}

\input{development}

\input{usecase}
\input{analysis}

\input{conclusion}

\input{acknowledgements}

\ifall
\fi

\bibliographystyle{IEEEtran}
\bibliography{IEEEabrv,ref}

\end{document}

%% file: abstract.tex
\begin{abstract}
Atomic Crosschain Transaction technology allows composable programming across private Ethereum blockchains. 
It allows for inter-contract and inter-blockchain function calls that are both synchronous and atomic: if one part fails, the whole call graph of function calls is rolled back. 
It is not based on existing techniques such as Hash Time Locked Contracts, relay chains, block header transfer, or trusted intermediaries. BLS Threshold Signatures are used to prove to validators on one blockchain that information came from another blockchain and that a majority of the validators of that blockchain agree on the information. 
Coordination Contracts are used to manage the state of a Crosschain Transaction and as a repository of Blockchain Public Keys. Dynamic code analysis and signed nested transactions are used together with live argument checking to ensure execution only occurs if the execution results in valid state changes. 
Contract Locking and Lockability enable atomic updates.

\end{abstract}

%% file: introduction.tex
\section{Introduction}
Atomic Crosschain Transactions \cite{robinson2019b} for permissioned Ethereum blockchains \cite{robinson2018a} allow for inter-contract and inter-blockchain function calls that are both synchronous and atomic. Atomic Crosschain Transactions are special nested Ethereum transactions that include additional fields to facilitate the atomic behaviour securely. The technology has been designed to shield application developers from the complexity of crosschain transactions by incorporating the required changes into the Ethereum Client software.

Fig.~\ref{fig:arch} shows a system of blockchains. Enterprises are indicated by the colour of lines between two blockchains. For example, the enterprise represented by the green lines has nodes on blockchains A, C, D, X, and Y; the enterprise represented by the blue lines has nodes on blockchains A, B, D, and X; and the enterprise represented by the red lines has nodes on all blockchains. Using Atomic Crosschain Transaction technology, an enterprise can create a crosschain transaction who's call graph can traverse any blockchain they have a node on. The Coordination Blockchains are used to hold the cross transaction state. All nodes on all blockchains processing a crosschain transaction need to be able to access the Coordination Blockchain referenced in a crosschain transaction.

\begin{figure}[b]
  \includegraphics[width=\linewidth]{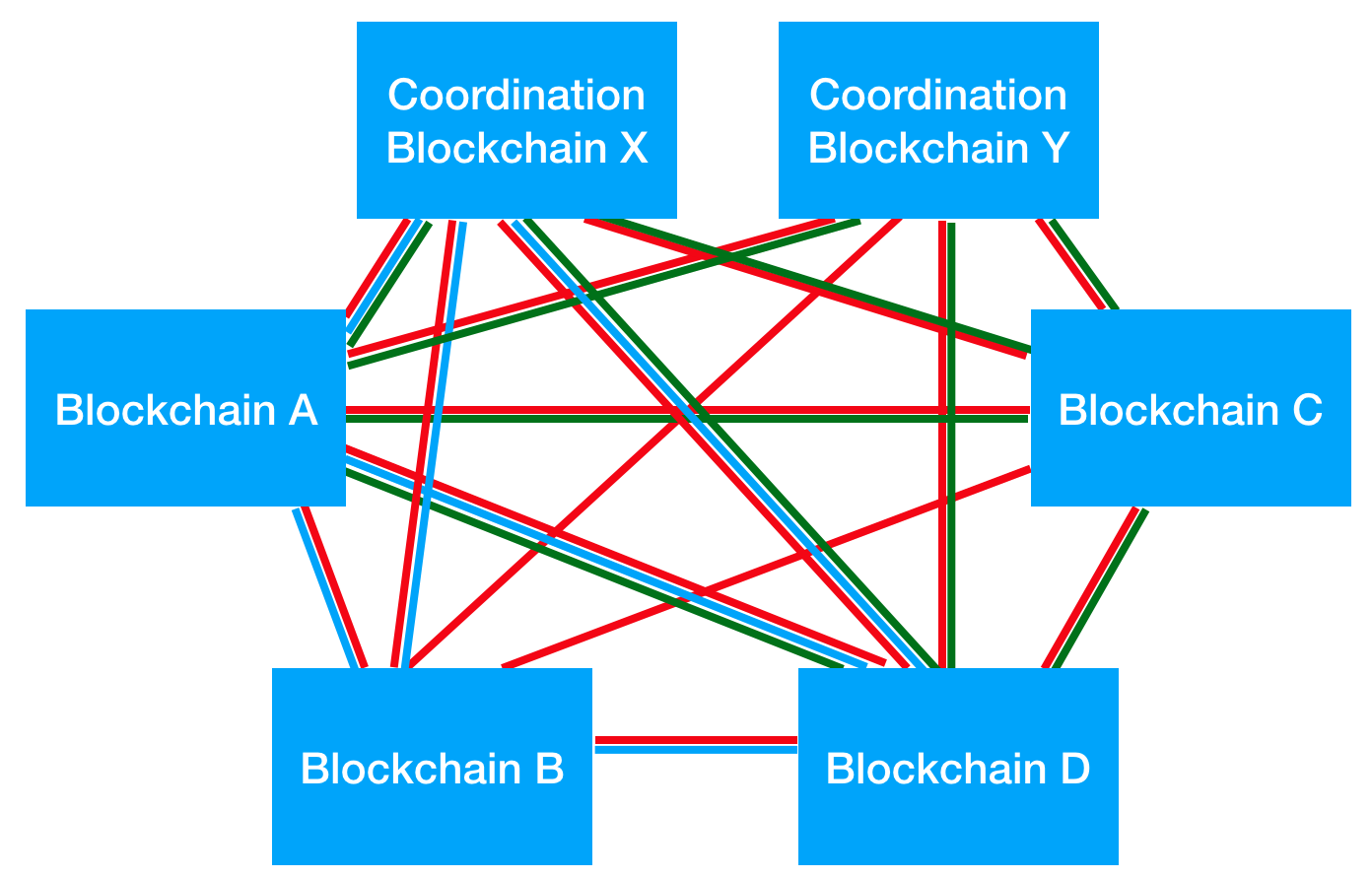}
  \caption{Cross-Blockchain Linkage}
  \label{fig:arch}
\end{figure}

Fig.~\ref{fig:functioncall} shows code fragments from two contracts on two blockchains. The function \texttt{funcA} in contract \texttt{ConA} on Private Blockchain A calls the function \texttt{funcB} in contract \texttt{ConB} on Private Blockchain B. Atomic Crosschain Transactions allow the crosschain function call to occur synchronously and atomically. By synchronously it is meant that function calls that return results can be executed and the results are immediately available and that all execution will complete and updates will be committed prior to the end of the Atomic Crosschain Transaction. By atomic it is meant that the updates are either applied on both blockchains, or the updates are discarded on both blockchains.

\begin{figure}
  \includegraphics[width=\linewidth]{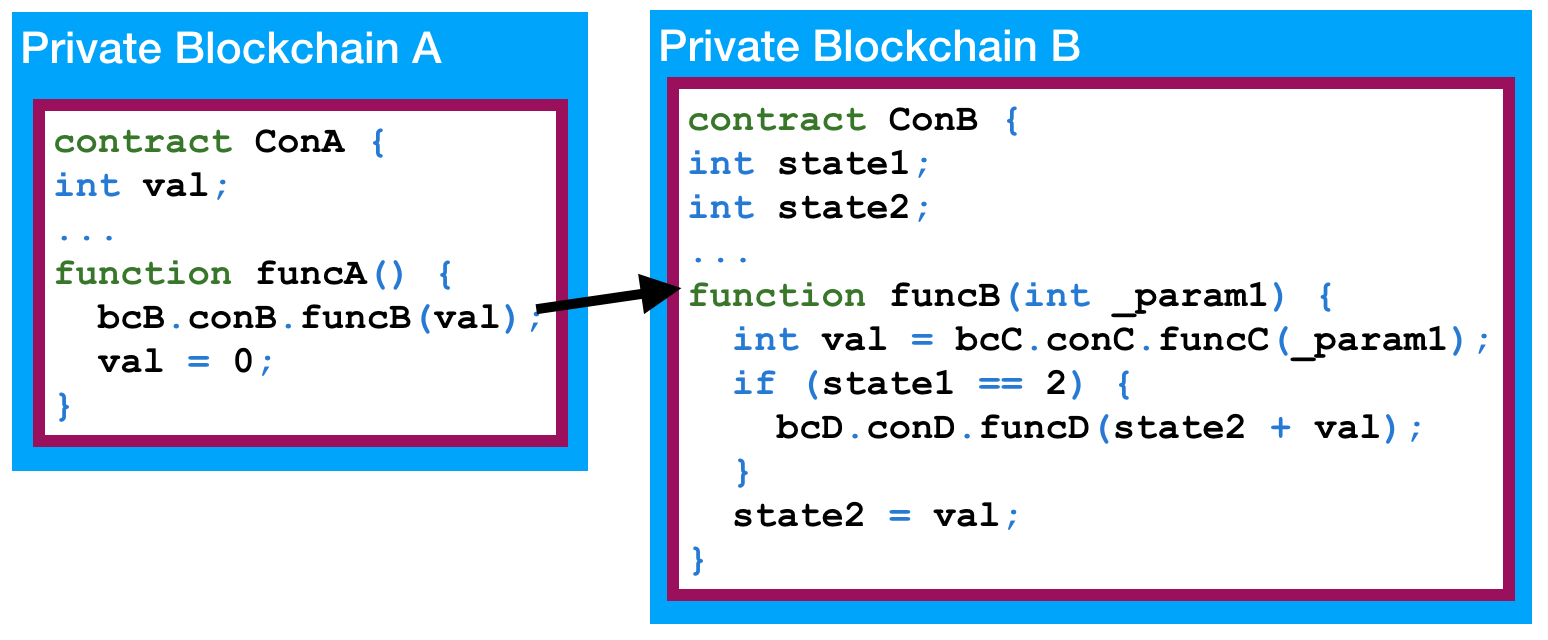}
  \caption{Function Calls across Blockchains}
  \label{fig:functioncall}
\end{figure}

Atomic Crosschain Transactions technology does not rely on Time Hash Lock Contracts \cite{hashtimelock}, relay chains \cite{cosmos2016}, block header relaying \cite{Clearmatics2018c,btc-relay}, or trusted intermediaries. It is a new technique which brings together BLS Threshold Signatures, nested transactions, and dynamic parameter checking against expected signed values, and a Coordination Blockchain as a holder of crosschain transaction state and as a time-out time reference.

Section~\ref{sec:components} \textit{Components} describes the concepts behind the technology. Section~\ref{sec:processing} \textit{Transaction Processing} walks through happy case and failure case transaction processing. Section~\ref{sec:development} \textit{Application Development} discusses considerations that should be considered when creating applications with this technology. Section~\ref{sec:analysis} \textit{Analysis} describes the Safety and Liveness properties of this technology.

Source code for this technology is available on github.com\cite{crosschain-github}.

%% file: components.tex
\section{Components}
\label{sec:components}
\subsection{Nested Transactions}
Atomic Crosschain Transactions are nested Ethereum transactions and views. Transactions are function calls that update state. Views are function calls that return a value but do not update state. Fig.~\ref{fig:nested1} shows an Externally Owned Account (EOA) calling a function \texttt{funcA} in contract \texttt{ConA} on blockchain \texttt{Private Blockchain A}. This function in turn calls function \texttt{funcB}, that in turn calls functions \texttt{funcC} and \texttt{funcD}, each on separate blockchains. The transaction submitted by the EOA is called the \textit{Originating Transaction}. The transactions that the Originating Transaction causes to be submitted are called Subordinate Transactions. Subordinate Views may also be triggered. In Fig.~\ref{fig:nested1}, a Subordinate View is used to call \texttt{funcC}. This function returns a value to \texttt{funcB}.

\begin{figure}
  \includegraphics[width=\linewidth]{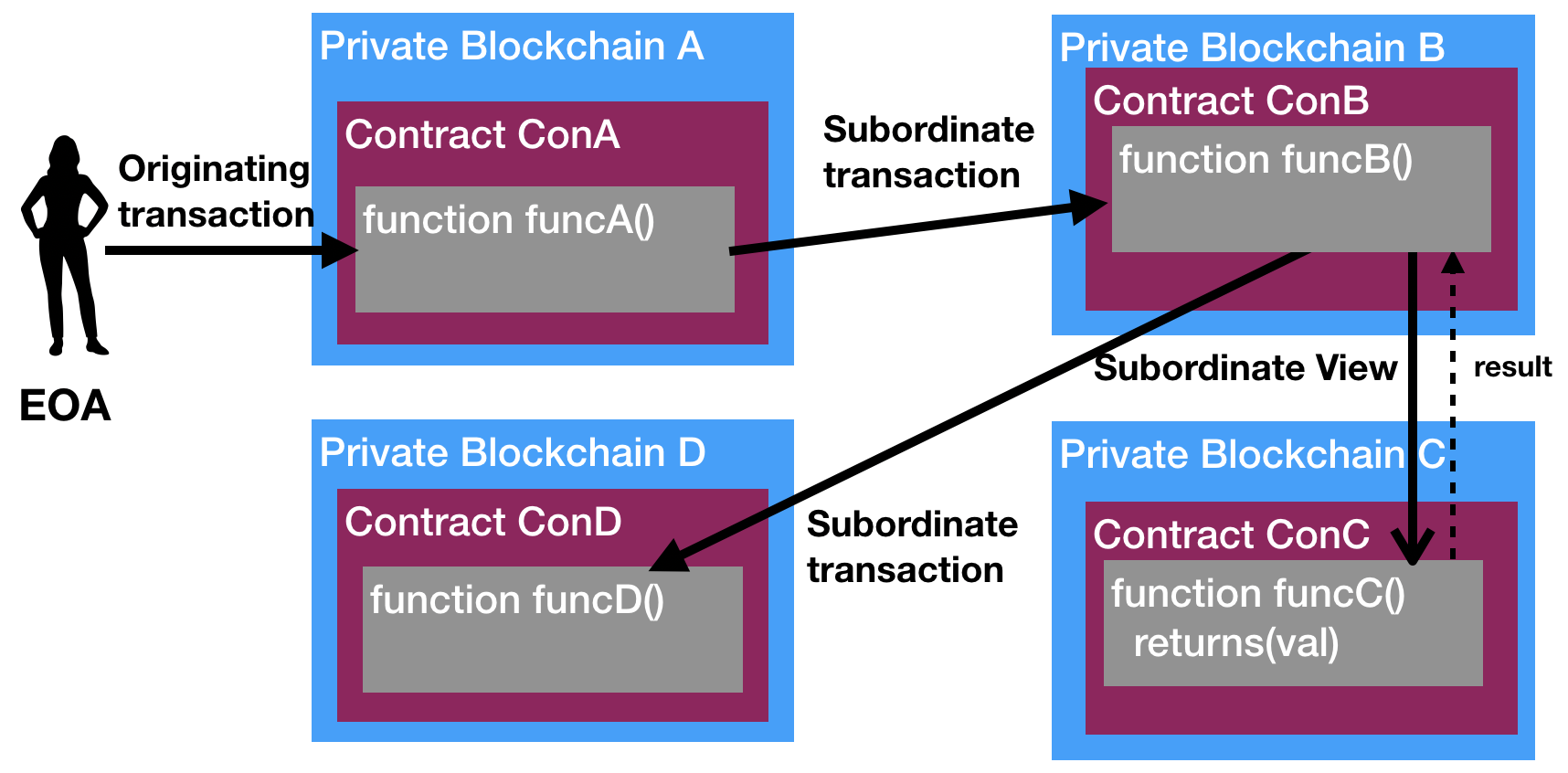}
  \caption{Originating Transaction containing Two Nested Subordinate Transactions and a Subordinate View}
  \label{fig:nested1}
\end{figure}

Fig.~\ref{fig:nested2} shows the nested structure of the Atomic Crosschain Transaction. The EOA user first creates the signed Subordinate View for \texttt{Private Blockchain C}, contract \texttt{ConC}, function \texttt{funcC} and the signed Subordinate Transaction for \texttt{Private Blockchain D}, contract \texttt{ConD}, function \texttt{funcD}. They then create the signed Subordinate Transaction for \texttt{Private Blockchain B}, contract \texttt{ConB}, function \texttt{funcB}, and include the signed Subordinate Transaction and View. Finally, they sign the Originating Transaction for \texttt{Private Blockchain A}, contract \texttt{ConA}, function \texttt{funcA}, including the signed Subordinate Transactions and View.  

\begin{figure}
  \centering
  \includegraphics[width=0.7\linewidth]{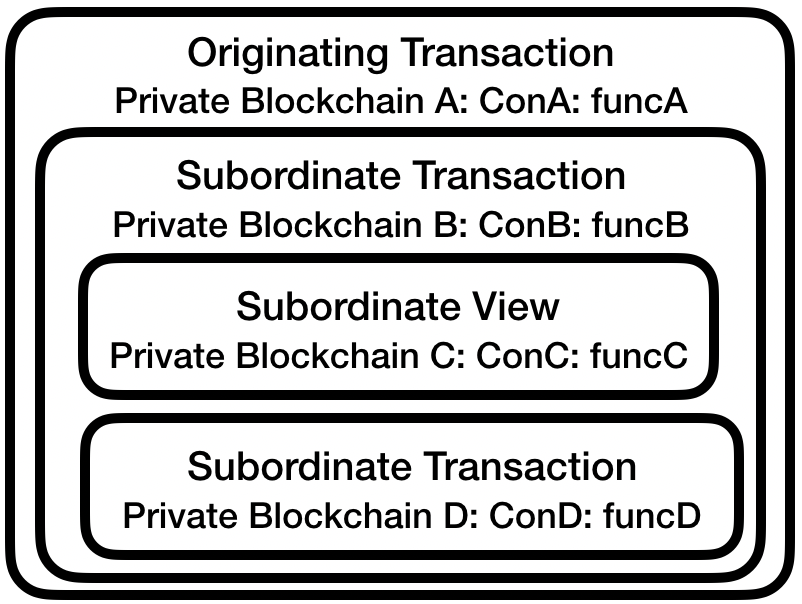}
  \caption{Nested Transactions and Views}
  \label{fig:nested2}
\end{figure}

\subsection{Actual and Expected Parameter Values}
When a function executes in the Ethereum Virtual Machine \cite{wood2016a} function parameter values and stored state combine to form the actual values of variables during execution. For example, consider \texttt{funcB} in contract \texttt{ConB} on \texttt{Private Blockchain B} in Fig.~\ref{fig:functioncall}. Assuming \texttt{\_param1} is \texttt{1}, \texttt{state1} is \texttt{2}, \texttt{state2} is \texttt{4}, and that \texttt{funcC} returns the value \texttt{6}, then function \texttt{funcC} will be called with the parameter value \texttt{1}, and function \texttt{funcD} will be called with the parameter value \texttt{10}. To execute this as part of a Crosschain Transaction, signed Subordinate Transactions and View need to be created with the appropriate parameter values. 

Execution of a transaction or view fails if the actual parameter values passed to a function does not match the values in the signed Subordinate Transaction or View. The parameter values in the signed Subordinate Transaction or View are the values the application developer expected to be passed to the function. The expected values can be determined at time of nested transaction creation using dynamic program analysis techniques. In particular, the dynamic analysis needs to consider program flow. For instance, if \texttt{state1} was \texttt{1}, then \texttt{funcD} would not be called, and no Subordinate Transaction should be included in the Crosschain Transaction for this function call. This ensures that Subordinate Transactions and Views are executed on the exact state that was intended at the time of creation of the crosschain transaction.

\subsection{Per-Node Transaction Processing}

When the EOA submits the Originating Transaction to a node, the node processes the transaction using the algorithm shown in Listing~\ref{listing:processing}. If the transaction includes any Subordinate Views, they are dispatched and their results are cached (Lines 1 to 3). The function is then executed (Lines 4 to 17). If a Subordinate Transaction function call is encountered, the node checks that the parameter values passed to the Subordinate Transaction function call match the parameter values in the signed Subordinate Transaction (Lines 6 to 8). If a Subordinate View function call is encountered, the node checks that the parameters passed to the Subordinate View function call match the parameter values in the signed Subordinate View (Lines 9 and 10). The cached values of the results of the Subordinate View function calls are then returned to the executing code (Line 11). If the execution has completed without error, then each of the signed Subordinate Transactions is submitted to a node on the appropriate blockchain (Nodes 18 to 20).

\begin{lstlisting}[
%  frame=single,
  basicstyle=\footnotesize\ttfamily,
  numbers=left,
stepnumber=1, 
  firstnumber=1,
  numberfirstline=true,
  numbersep=5pt,    
  xleftmargin=0.5cm,
  morekeywords={msg},
  label=listing:processing,
  caption=Originating or Subordinate Transaction Processing
]
For All Subordinate Views {
  Dispatch Subordinate Views & cache results
}
Trial Execution of Function Call {
  While Executing Code {
    If Subordinate Transaction function called {
      check expected & actual parameters match.
    } 
    Else If Subordinate View function is called {
      check expected & actual parameters match
      return cached results to code
    } 
    Else {
      Execute Code As Usual
    }
  }
}
For All Subordinate Transactions {
  Submit Subordinate Transactions
}
\end{lstlisting}

\subsection{Blockchain Signing and Threshold Signatures}
BLS Threshold Signatures \cite{bls-threshold, bls-threshold-youtube} combines the ideas of threshold cryptography \cite{shamir1979} with Boneh-Lynn-Shacham(BLS) signatures \cite{bls2004}, and uses a Pedersen commitment scheme \cite{ped1991} to ensure verifiable secret sharing. The scheme allows any \texttt{M} validator nodes of the total \texttt{N} validator nodes on a blockchain to sign messages in a distributed way such that the private key shares do not need to be assembled to create a signature. Each validator node creates a signature share by signing the message using their private key share. Any \texttt{M} of the total \texttt{N} signature shares can be combined to create a valid signature. Importantly, the signature contains no information about which nodes signed, or what the threshold number of signatures (\texttt{M}) needed to create the signature is.

The Atomic Crosschain Transaction system uses BLS Threshold Signatures to prove that information came from a specific blockchain. For example, in Fig.~\ref{fig:nested1}, nodes on \texttt{Private Blockchain B} can be certain of results returned by a node on \texttt{Private Blockchain C} for the function call to \texttt{funcC}, as the results are threshold signed by the validator nodes on \texttt{Private Blockchain C}. Similarly, validator nodes on \texttt{Private Blockchain A} can be certain that validator nodes on \texttt{Private Blockchain B} have mined the Subordinate Transaction, locked contract \texttt{ConB} and are holding the updated state as a provisional update because validator nodes sign a \textit{Subordinate Transaction Ready} message indicating that the Subordinate Transaction is ready to be committed.

\subsection{Multichain Nodes}
A Multichain Node is a logical grouping of one or more blockchain validator nodes, where each node is on a different blockchain. The blockchain nodes operate together to allow Crosschain Transactions. The Multichain Node on which the transaction is submitted must have Validator Nodes on all of the blockchains on which the Originating Transaction and Subordinate Transactions and Views take place. 

Fig.~\ref{fig:multichain} shows four enterprises that have validator nodes on \texttt{Private Blockchain A} to \texttt{Private Blockchain D}. Alice who works in Enterprise 1 can submit Atomic Crosschain Transactions that span \texttt{Private Blockchain A} to \texttt{Private Blockchain D} as Enterprise 1 has a Multichain Node that includes validator nodes on each blockchain. However, Bob who works in Enterprise 4 can only submit Atomic Crosschain Transactions that span \texttt{Private Blockchain B} and \texttt{Private Blockchain C} as Enterprise 4 only has validator nodes on \texttt{Private Blockchain B} and \texttt{Private Blockchain C}.

\begin{figure}
  \includegraphics[width=\linewidth]{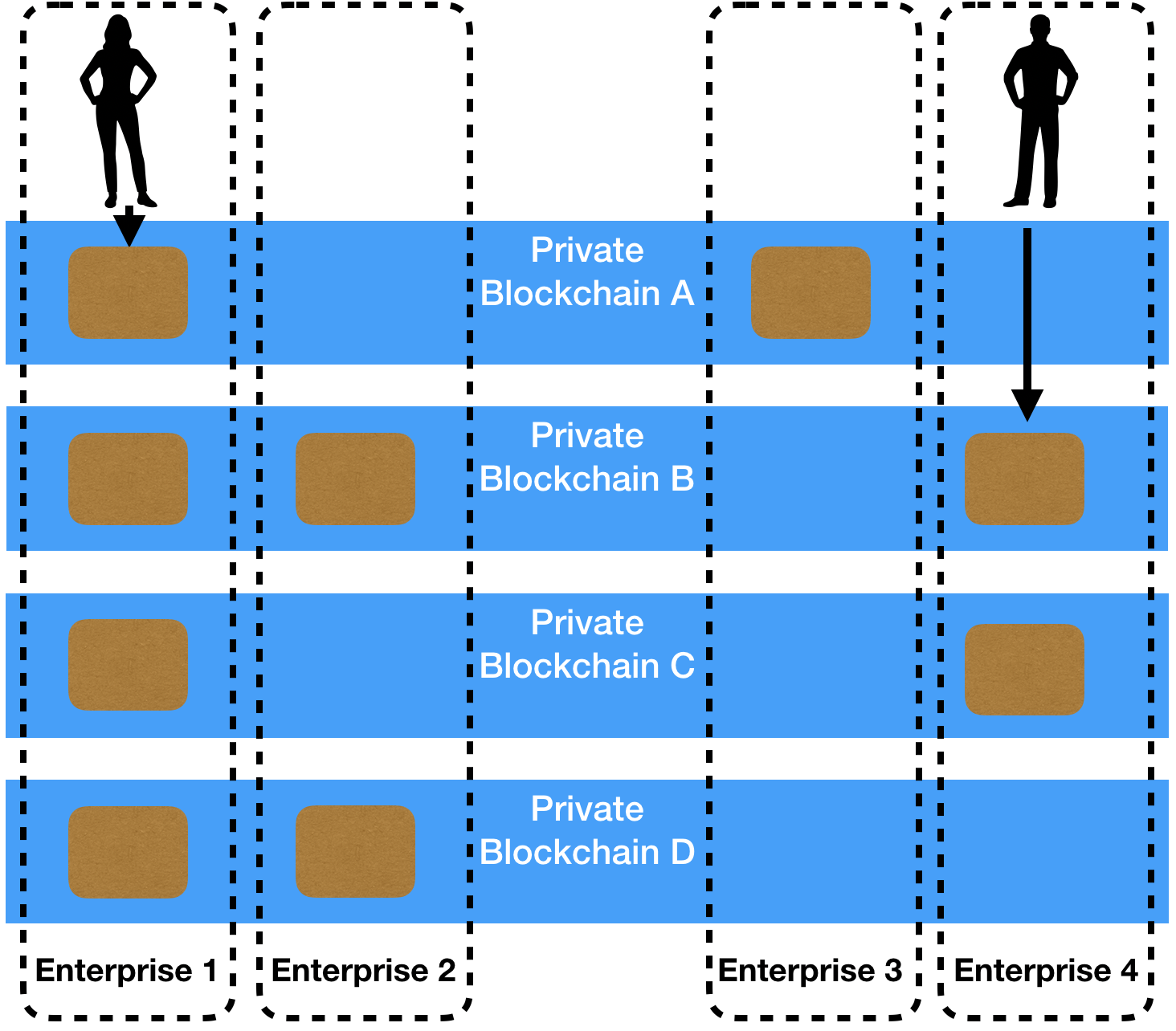}
  \caption{Multichain Nodes}
  \label{fig:multichain}
\end{figure}

\subsection{Crosschain Coordination}
\textit{Crosschain Coordination Contracts} exist on \textit{Coordination Blockchains}. They allow validator nodes to determine whether the provisional state updates related to the Originating Transaction and Subordinate Transactions should be committed or discarded. The contract is also used to determine a common time-out for all blockchains, and as a repository of Blockchain Public Keys. 

When a user creates a Crosschain Transaction, they specify the Coordination Blockchain and Crosschain Coordination Contract to be used for the transaction, and the time-out for the transaction in terms of a block number on the Coordination Blockchain. The validator node that they submit the Originating Transaction to (the \textit{Originating Node}) works with other validator nodes on the blockchain to sign a \textit{Crosschain Transaction Start} message. This message is submitted to the Crosschain Coordination Contract to indicate to all nodes on all blockchains that the Crosschain Transaction has commenced. 

When the Originating Node has received \textit{Subordinate Transaction Ready} messages for all Subordinate Transactions, it works with other validator nodes to create a \textit{Crosschain Transaction Commit} message. This message is submitted to the Crosschain Coordination Contract to indicate to all nodes on all blockchains that the Crosschain Transaction has completed and all provisional updates should be committed. If an error is detected, then a \textit{Crosschain Transaction Ignore} message is created and submitted to the Crosschain Coordination Contract to indicate to all nodes on all blockchains that the Crosschain Transaction has failed and all provisional updates should be discarded. Similarly, if the transaction times-out, all provisional updates will be discarded.

\subsection{Contract Locking and Provisional State Updates}
When a contract is first deployed it is marked as a Lockable Contract or a Nonlockable Contract. A Nonlockable Contract, the default, is one which can not be locked. When a node attempts to update the state of a contract given an Originating or Subordinate Transaction, it checks whether the contract is \textit{Lockable} and whether it is \textit{locked}. The transaction fails if the contract is Nonlockable or if the contract is Lockable but is locked.

Figure \ref{fig:contractlockingstates} shows the locking state transitions for a contract. The Crosschain Coordination Contract will be in \textit{Started} state. The act of mining an Originating Transaction or Subordinate Transaction and including it in a blockchain locks the contract. The contract is unlocked when the Crosschain Coordination Contract is in the \textit{Committed} or \textit{Ignored} state, or when the block number on the Coordination Blockchain is greater than the Transaction Timeout Block Number. The Crosschain Coordination Contract will change from the \textit{Started} state to the \textit{Committed} state when a valid Crosschain Transaction Commit message is submitted to it, and it will change from the \textit{Started} state to the \textit{Ignored} state when a valid Crosschain Transaction Ignore message is submitted to it. 
\begin{figure}
  \includegraphics[width=\linewidth]{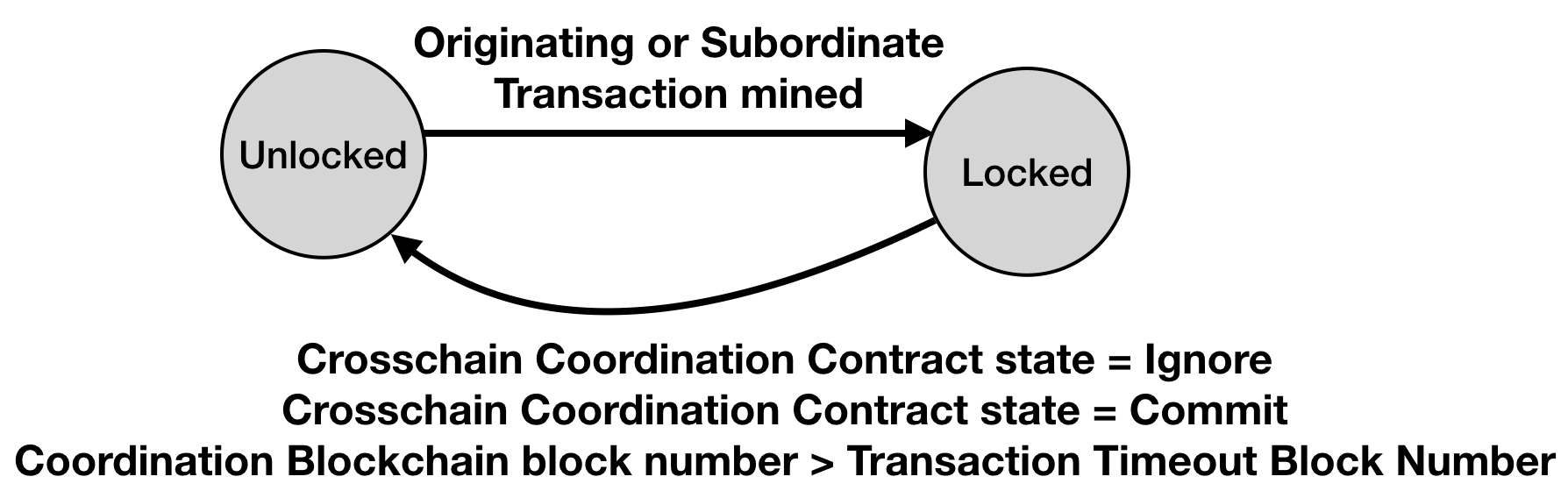}
  \caption{Contract Locking States}
  \label{fig:contractlockingstates}
\end{figure}

Ordinarily, all nodes will receive a Signalling Transaction indicating that they should check the Crosschain Coordination Contract when the contract can be unlocked. When a node first processes a transaction, it will set a local timer which should expire when the Transaction Timeout Block Number is exceeded. If the node has not received the message by the time the local timer expires, the node checks the Crosschain Coordination Contract to see if the Transaction Timeout Block Number has been exceeded and whether the updates should be committed or ignored.

\subsection{Crosschain Transaction Fields}
Originating Transactions, Subordinate Transactions, and Subordinate Views contain the fields shown in Table~\ref{table:fields}. Some of the information in the standard Ethereum transaction fields are exposed to blockchain application contract code, such as the \texttt{value} field via the Solidity code \texttt{msg.value}. The new extended crosschain transaction fields are made available to blockchain application contract code via a precompile contract.

\begin{table}
  \centering
    \begin{tabular}{| l | l |}
    \hline
    Field & Description  \\
       \hline
       \hline
        \multicolumn{2}{|l|}{Standard Ethereum Transaction Fields} \\
       \hline
Nonce     & Per-account, per-blockchain transaction number. \\
       \hline
GasPrice & Amount offered to pay for gas for the transaction.\\
       \hline
GasLimit  &  Maximum gas which can be used by the transaction.\\
       \hline
To            & Address of the account to send the value to, or the  \\
                & address of a contract to call.\\
       \hline
Value       & Amount of Ether to transfer.\\
       \hline
Data         & Encoded function signature and parameter values.\\
       \hline
V              & Part of the transaction digital signature \& blockchain \\
                & identifier this transaction must execute on.\\
       \hline
R              & Part of the transaction digital signature.\\
       \hline
S              & Part of the transaction digital signature.\\
       \hline
       \hline
        \multicolumn{2}{|l|}{Additional Crosschain Transaction Fields} \\
       \hline
    Type     & Type of crosschain transaction (e.g. Originating \\
                 & Transaction) \\
    \hline
Coordination & Blockchain identifier of Coordination Blockchain to  \\
Blockchain Id & use for this transaction.\\
       \hline
Crosschain  & Address of the Crosschain Coordination Contract \\
Coordination &  to use for this transaction. \\
Contract       & \\
       \hline
Crosschain   & Coordination Blockchain block number when this\\ 
Transaction  & transaction will time out. \\
Time-out      & \\
       \hline
Crosschain  &  Identifies this crosschain transaction. \\
Transaction Id & \\
       \hline
Originating  & Blockchain identifier of the blockchain the  \\
Blockchain Id &  Originating Node is on. \\
       \hline
From             & Blockchain identifier of the blockchain that the   \\
Blockchain Id &  function call executed on that resulted in this \\
                       & Subordinate Transaction or View being submitted. \\
       \hline
From            & \textit{To} address from the transaction or view that resulted \\
Address        & in this Subordinate  Transaction or View.\\
       \hline
Subordinates & List of Subordinate Transactions and Subordinate  \\
                       & Views that are called directly from this transaction \\
                       & or view. \\
       \hline
  \end{tabular}
  \caption{Crosschain Transaction Fields}
  \label{table:fields}
\end{table}

All nodes that process the transaction check that the \texttt{Coordination Blockchain Id}, \texttt{Crosschain Co- ordination Contract}, \texttt{Crosschain Transaction Time-out}, \texttt{Crosschain Transaction Id}, and \texttt{Originating Blockchain Id} are consistent across the transaction or view they are processing, and the nested Subordinate Transactions and Views. The nodes also check that the \texttt{To} address and \texttt{From Address}, and the blockchain identifier obtained from the \texttt{V} field and the \texttt{From Blockchain Id} match across transactions and views.

The \texttt{To} address is the address of the contract containing the function called on a blockchain. For example, the function (f1) in contract (c1) could call a function (f2) in another contract (c2) on the same blockchain (b1). The second contract (c2) could call a function (f3) in a contract (c3) on another blockchain (b2) via a Subordinate Transaction. The \texttt{From Address} of the Subordinate Transaction will match the \texttt{To} address of the transaction on the first blockchain (b1). This will be the address first contract (c1). It will however, not match the address of the second contract (c2), which is the function that caused the Subordinate Transaction to be triggered.

%% file: processing.tex
\section{Transaction Processing}
\label{sec:processing}
\subsection{Happy Case}
\begin{figure}
  \includegraphics[width=\linewidth]{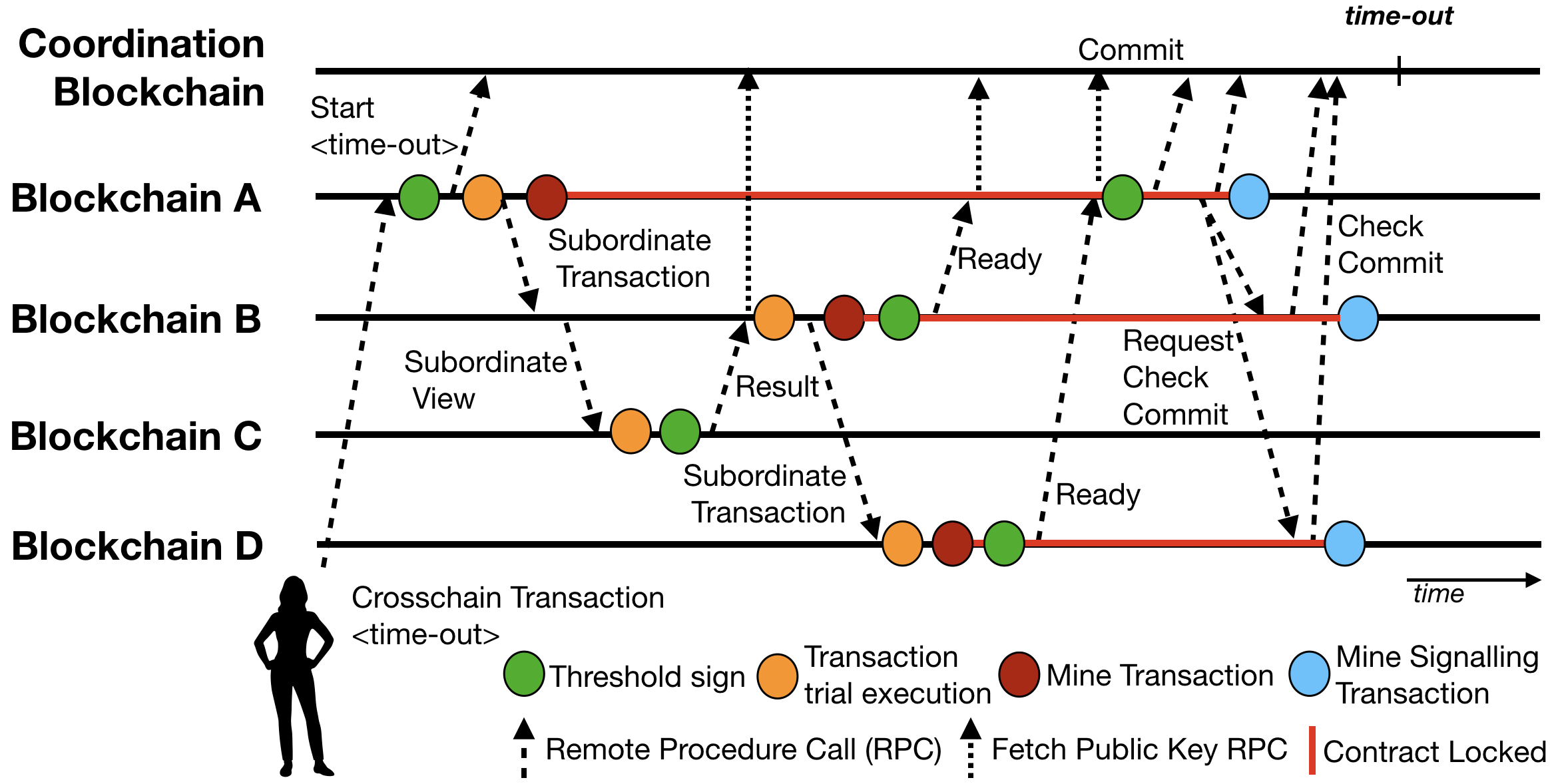}
  \caption{Sequence Diagram}
  \label{fig:chapter7:sequence}
\end{figure}

Fig.~\ref{fig:chapter7:sequence} shows a sequence diagram showing an example sequence of execution, based on the calls in Fig.~\ref{fig:nested1}. The sequence starts with an Externally Owned Account (EOA) submitting a Crosschain Transaction to a node on Blockchain A. This nested transaction contains the Originating Transaction, and all Subordinate Transactions and Views, plus a time-out in terms of Coordination Blockchain block numbers. The node is known as the Coordinating Node for this transaction. 

The nodes on Blockchain A cooperate to threshold sign the \textit{Start} message. The signed \textit{Start} message is submitted to the Coordination Contract on the Coordination Blockchain specified in the Originating Transaction (and all Subordinate Transactions and Views). The Coordination Contract on the Coordination Blockchain accepts the \textit{Start} message if its signature can be verified with the public key registered for Blockchain A. The \textit{Start} message contains the time-out block number, which the Coordinating Contract registers when it accepts the \textit{Start} message. Once the \textit{Start} is registered, the Crosschain Transaction is in \textit{Start} state.

The Coordinating Node then \textit{Trial Executes} the Originating Transaction. Assuming the code executes correctly, the Subordinate Transaction for Blockchain B is dispatched. The Originating Transaction is then submitted to Blockchain A's transaction pool. When the transaction is included in a block, the contract's updated state is stored as provisional state, and the contract state is locked.

The Subordinate Transaction for Blockchain B contains a Subordinate View for Blockchain C. The Subordinate View is dispatched to a node on Blockchain C. The node on Blockchain C completes a \textit{Trial Execution} of the Subordinate View at a certain block number. The node then cooperates with other Blockchain C nodes to have the result of the Subordinate View at the specific block number threshold signed. The node returns the signed result to the calling node on Blockchain B. 

The nodes on Blockchain B that do not contain the correct version of Blockchain C's public key fetch it from the Coordination Contract. The node \textit{Trial Executes} the Subordinate Transaction, returning the value from the Subordinate View to the code when the function corresponding to the Subordinate View is called. Assuming the \textit{Trial Execution} executes correctly, then the Subordinate Transaction to be executed on Blockchain D is dispatched to a node on Blockchain D. The Subordinate Transaction for Blockchain B is submitted to Blockchain B's transaction pool, with the signed Subordinate View result from Blockchain C attached. When the transaction is included in a block, the contract's updated state is stored as provisional state, and the contract state is locked. The nodes on Blockchain B then cooperate to threshold sign a \textit{Ready} message, indicating that the blockchain is ready to commit the Atomic Crosschain Transaction. The message is submitted to the Coordinating Node.

The Subordinate Transaction for Blockchain D is trial executed. There are no Subordinate Transactions or Views for the transaction that need to be dispatched. The transaction is mined. When the transaction has been included in a block, the updated state is stored in provisional state, the contract is locked, and the nodes on Blockchain D cooperate to threshold sign a \textit{Ready} message, indicating that the blockchain is ready to commit the Atomic Crosschain Transaction. The message is submitted to the Coordinating Node.

When the Coordinating Node (on Blockchain A) has received \textit{Ready} messages for Blockchain B and D, and the Originating Transaction has been mined, then the nodes on Blockchain A then cooperate to threshold sign a \textit{Commit} message. This message is uploaded to the Coordination Contract on the Coordination Blockchain. Assuming this message is uploaded prior to the time-out, the Crosschain Transaction is then in the commit state. 

The Coordinating Node requests that nodes on other Blockchains commit a Signalling Transaction. When a node on a blockchain encounters a Signalling Transaction for a certain Crosschain Transaction, it checks the Coordination Contract for the Crosschain Transaction to check to see if the transaction should be committed or ignored. If the transaction is to be committed, then all contract provisional state related to the transaction is converted to normal state, unlocking the contracts. If the transaction is to be ignored, then the contract provisional state is discarded and the contract state is updated, unlocking the contracts.

\subsection{Failure Cases}
If an error is detected then nodes on Blockchain A can cooperate to create a threshold signed \textit{Ignore} message which is uploaded to the Coordination Contract to indicate to all nodes that the Crosschain Transaction should be abandoned, all provisional state updates should be discarded and contracts should be unlocked. If an \textit{Ignore} message can not be created, or if the Crosschain Transaction times-out, then the state returned by the Coordination Contract indicates the Crosschain Transaction should be abandoned.

When nodes on each blockchain first become a part of a Crosschain Transaction, they set a timer to expire when they expect the Crosschain Transaction to time-out, plus a random additional wait period. The timer is cancelled when a Signalling Transaction is received. If a node's timer expires prior to receiving a Signalling Transaction, the node checks the Coordination Contract and then creates and submits a Signalling Transaction to indicate to other nodes that the updates should be committed or ignored, and that the contract should be unlocked. The additional random wait period is used so that all nodes on the blockchain don't simultaneously send Signalling Transactions.

%% file: development.tex
\section{Application Development}
\label{sec:development}
\subsection{Design for Locking}
Contracts that have state updated as part of a crosschain transaction are locked. Calls to those locked contracts will fail until they are unlocked. Traditional contracts that hold registries of information, which may need to be accessed by many users at the same time, such as ERC 20 contracts, need to be refactored to work successfully in a crosschain scenario. They should be restructured such that there is a non-lockable \textit{router} contract that determines the appropriate lockable \textit{item} contract. 

Fig~\ref{fig:erc20} shows how an ERC 20 contract could be refactored. Each account could have one or more lockable ERC 20 Payment Slot contracts. The non-lockable ERC 20 Router contract's code could be written such that transfers are sent from an account's payment slot contract that isn't locked which has adequate funds to a payment slot for the destination account that is not currently locked. In this way, multiple payments may be executed in parallel.
\begin{figure}
  \centering
  \includegraphics[width=0.7\linewidth]{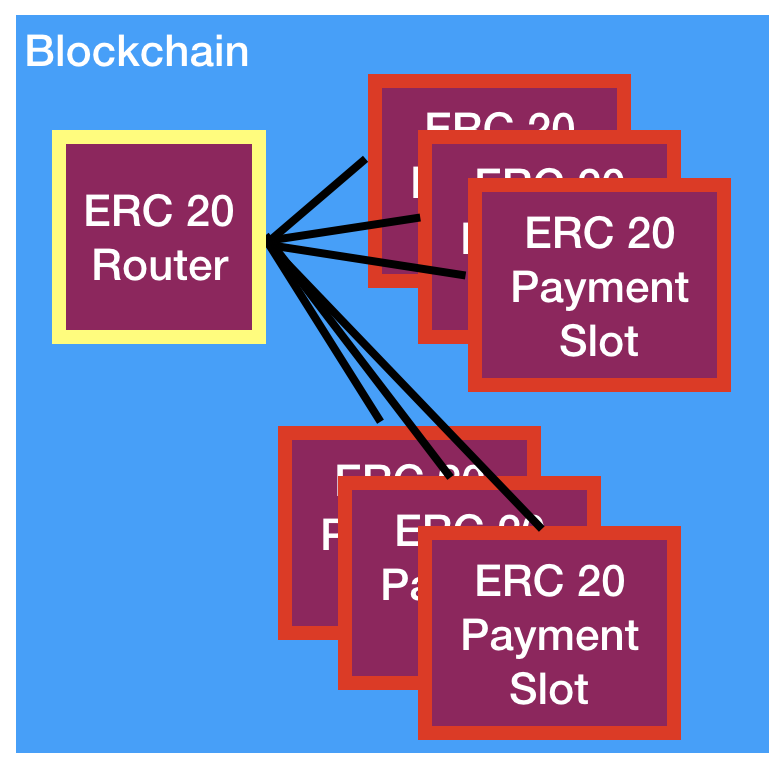}
  \caption{ERC 20 Contract for Crosschain}
  \label{fig:erc20}
\end{figure}

\subsection{Application Level Authentication}
As with traditional Ethereum transactions, the type of application level authentication required for a Crosschain Transaction will be application dependent. 
\subsubsection{No Authentication}
Many functions will need no authentication at all. That is, functions can be designed such that it is safe to execute a transaction or return results of a view to any caller who is able to access the function. 

\subsubsection{Using \texttt{msg.sender} or \texttt{tx.origin}}
From the perspective of each Originating Transaction, Subordinate Transaction or View, \texttt{msg.sender} and \texttt{tx.origin} operate in the same way as a standard Ethereum transaction. That is, if an EOA submitted a transaction that called a function in contract A that then called a function in contract B on the same blockchain, \texttt{msg.sender} for contract B is contract A, and is the EOA for contract A. In both cases \texttt{tx.origin} would be the EOA. In the context of a node processing an Originating Transaction, Subordinate Transaction or View, for the purposes of  \texttt{msg.sender} and \texttt{tx.origin}, the transaction or view appears as a separately signed transaction. Given the similarities with standard Ethereum, \texttt{msg.sender} and \texttt{tx.origin} could be used in the same way as standard Ethereum to authenticate which EOA or contract on the same blockchain called a function call using code similar to that shown in Listing~\ref{listing:example-one-blockchain}. 
\begin{lstlisting}[
%  frame=single,
  basicstyle=\footnotesize\ttfamily,
  numbers=left,
stepnumber=1, 
  firstnumber=1,
  numberfirstline=true,
  numbersep=5pt,    
  xleftmargin=0.5cm,
  morekeywords={function, external, require, sender, call, msg},
  label=listing:example-one-blockchain,
  caption=Checking msg.sender
]
function receiver() external {
  require(msg.sender == authorisedAddress);
  ...
}
\end{lstlisting}

A key difference between standard Ethereum views and Subordinate Views is that Subordinate Views are signed. As such, the variables \texttt{msg.sender} and \texttt{tx.origin} can be used within Subordinate Views, whereas they are not set in the context of normal Ethereum views (except for the case of \texttt{msg.sender} when one contract calls another contract). 

\subsubsection{From Blockchain Id, From Address, and Originating Blockchain Id}
If a contract needs to only respond to calls from a certain contract on a certain blockchain, then the code in Listing~\ref{listing:example-auth} should be used. The code checks that the \texttt{From Blockchain Id} and \texttt{From Address} match the authorised blockchain and address, and checks that the blockchains represented by \texttt{From Blockchain Id} and \texttt{Originating Blockchain Id} are semi-trusted. By semi-trusted it is meant that fewer than \texttt{M} validators operating the blockchain are Byzantine. Note that this scenario implies the contract should allow for any \texttt{msg.sender} and \texttt{tx.origin}. 
\begin{lstlisting}[
%  frame=single,
  basicstyle=\footnotesize\ttfamily,
  numbers=left,
stepnumber=1, 
  firstnumber=1,
  numberfirstline=true,
  numbersep=5pt,    
  xleftmargin=0.5cm,
  morekeywords={address, function, uint256, require, sender, call, msg},
  label=listing:example-auth,
  caption=Crosschain Application Authentication
]
function receiver() external {
  address fromAddr = infoPrecompile(FROM_ADDR);
  uint256 fromBcId = infoPrecompile(FROM_BCID);
  uint256 origBcId = infoPrecompile(ORIG_BCID);
  require(fromAddr == authorisedFromAddress);
  require(fromBcId == authorisedFromBcId);
  require(origBcId == authorisedOrigBcId);
  ...
}
\end{lstlisting}

%% file: usecase.tex
\section{Usage Scenarios}
The following sections introduces example scenarios, based on real world use cases, to illustrate how the Atomic Crosschain Transaction technology could be leveraged to achieve the desired outcome(s).

\subsection{Travel Agent Use Case}
This use case demonstrates how the well known `Hotel and Train' problem can be solved. In this scenario, the travel agent needs to ensure the atomicity of the combined booking transaction. In other words, the travel agent needs to ensure that they either book both the hotel room and the train seat, or neither, so that they avoid the situation where a hotel room is successfully booked but the train reservation fails, or vice versa. There are several private blockchains involved: the travel agency runs a private blockchain, and each hotel and train travel company also maintains its own private blockchain. 

Fig.~\ref{fig:TravelAgentUseCase01} illustrates the crosschain function calls to book a hotel room and reserve a seat on a train. Travel Agent A purchases ERC20 tokens on the hotel and train blockchains, which they can then use to pay for accommodation and travel. As can be seen from this diagram, atomic behaviour is mainly enforced via contract locking, so that the lock remains in place until both booking transactions have been signalled as having completed successfully. 

\begin{figure} [h]
\includegraphics[width=\columnwidth]{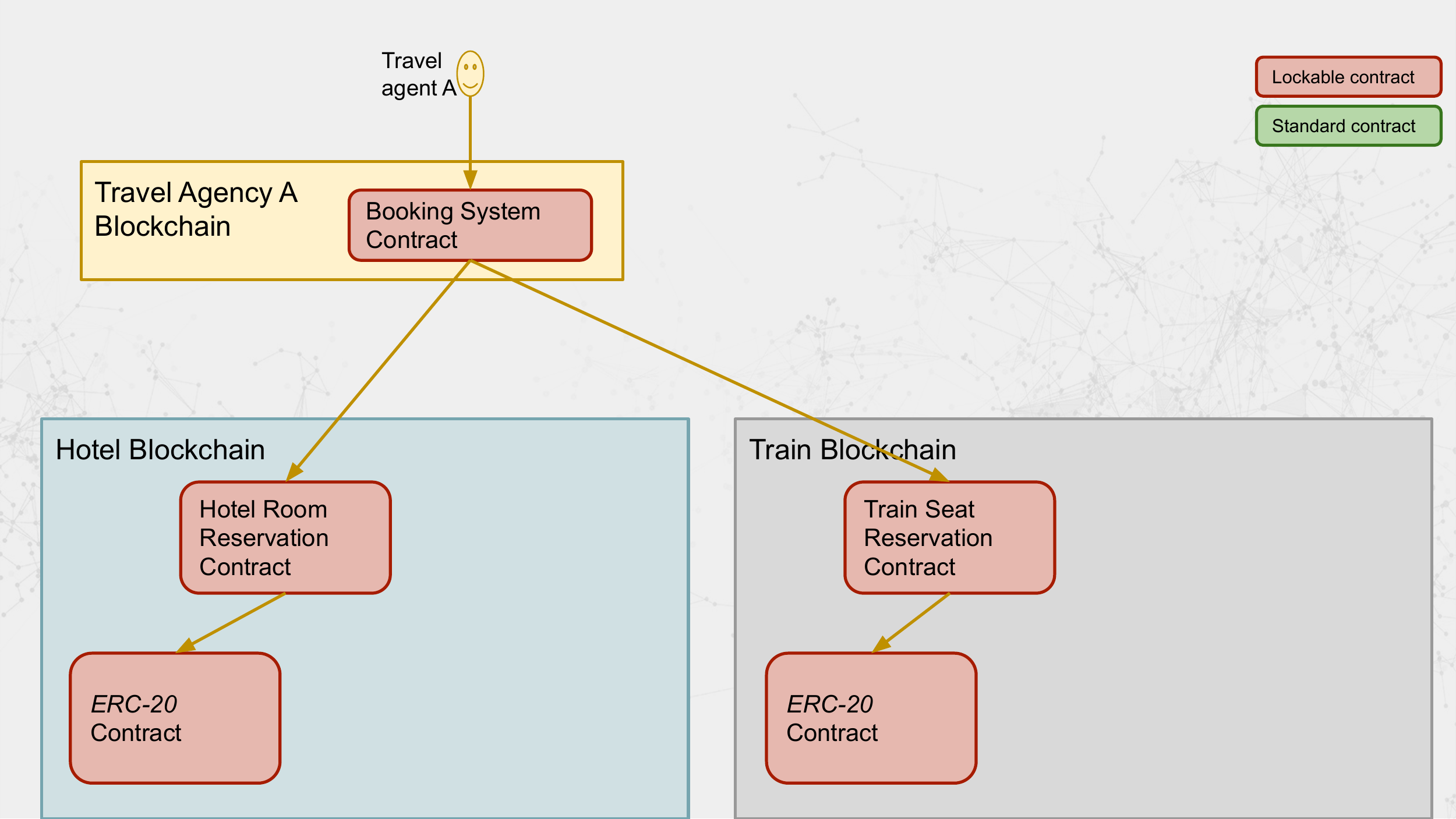}       
\caption{Travel Agent Use Case - One Agent}
\label{fig:TravelAgentUseCase01}
\end{figure}

This simplified version of the problem fails as soon as there is more than one travel agent trying to book the same hotel and/or train. For example, when the originating transaction from Travel Agent B initiates the subordinate crosschain function call to the hotel reservation contract (Fig.~\ref{fig:TravelAgentUseCase02}) before the previous crosschain function call graph from Travel Agent A had completed, the call will fail as the hotel reservation contract will still be locked. Although no funds would have been expended by Travel Agent B, this is clearly causing a bottleneck.

\begin{figure} [h]     
\includegraphics[width=\columnwidth]{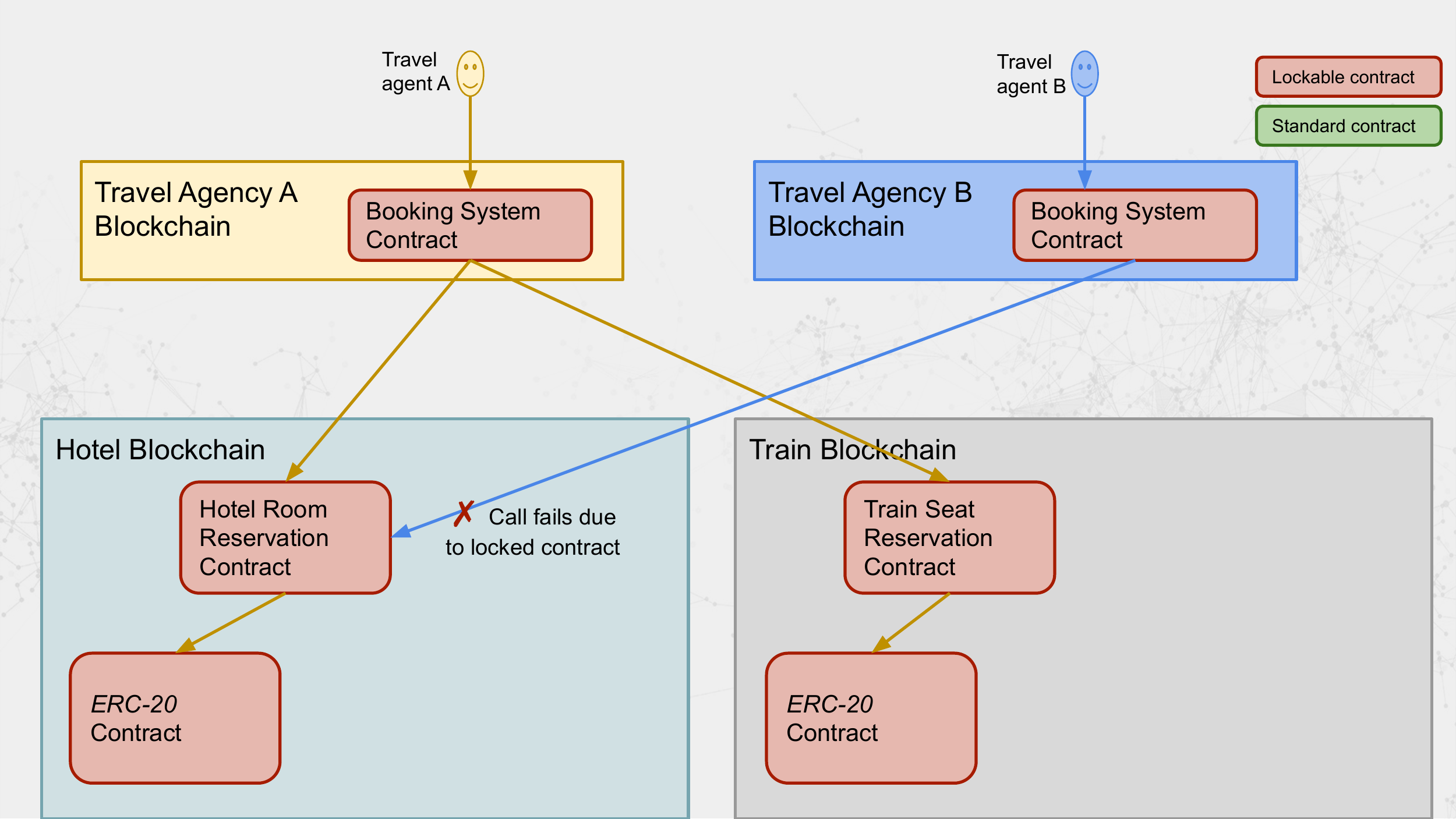}       
\caption{Travel Agent Use Case - Locking failure}
\label{fig:TravelAgentUseCase02}
\end{figure}

However, by adjusting the design of the `Hotel and Train' reservation system and its use of the Atomic Crosschain Transaction technology, we can leverage this functionality without the bottleneck shown in Fig.~\ref{fig:TravelAgentUseCase02}, as well as avoiding deadlocks. One approach is to introduce non-lockable router contracts, the green rounded rectangles depicted in Fig.~\ref{fig:TravelAgentUseCase03}, for both the hotel and train travel companies. These reservation router contracts are aware of every room's availability and status (provisional booking, booked, free) and similarly the reservation router contract on the train travel company blockchain. Moreover there are modified ERC20 contracts on both the hotel and train blockchains so that the travel agencies can pay for the hotel room and train seat. These modified ERC 20 contracts consist of a router contract and payment slot contracts, where each account has one or more payment slots. Following the crosschain function call graph illustrated in Figure~\ref{fig:TravelAgentUseCase03} it is evident that two travel agents could now simultaneously book hotel rooms and reserve train seats, by leveraging the Atomic Crosschain technology described in this paper. 

\begin{figure} [h]     
\includegraphics[width=\columnwidth]{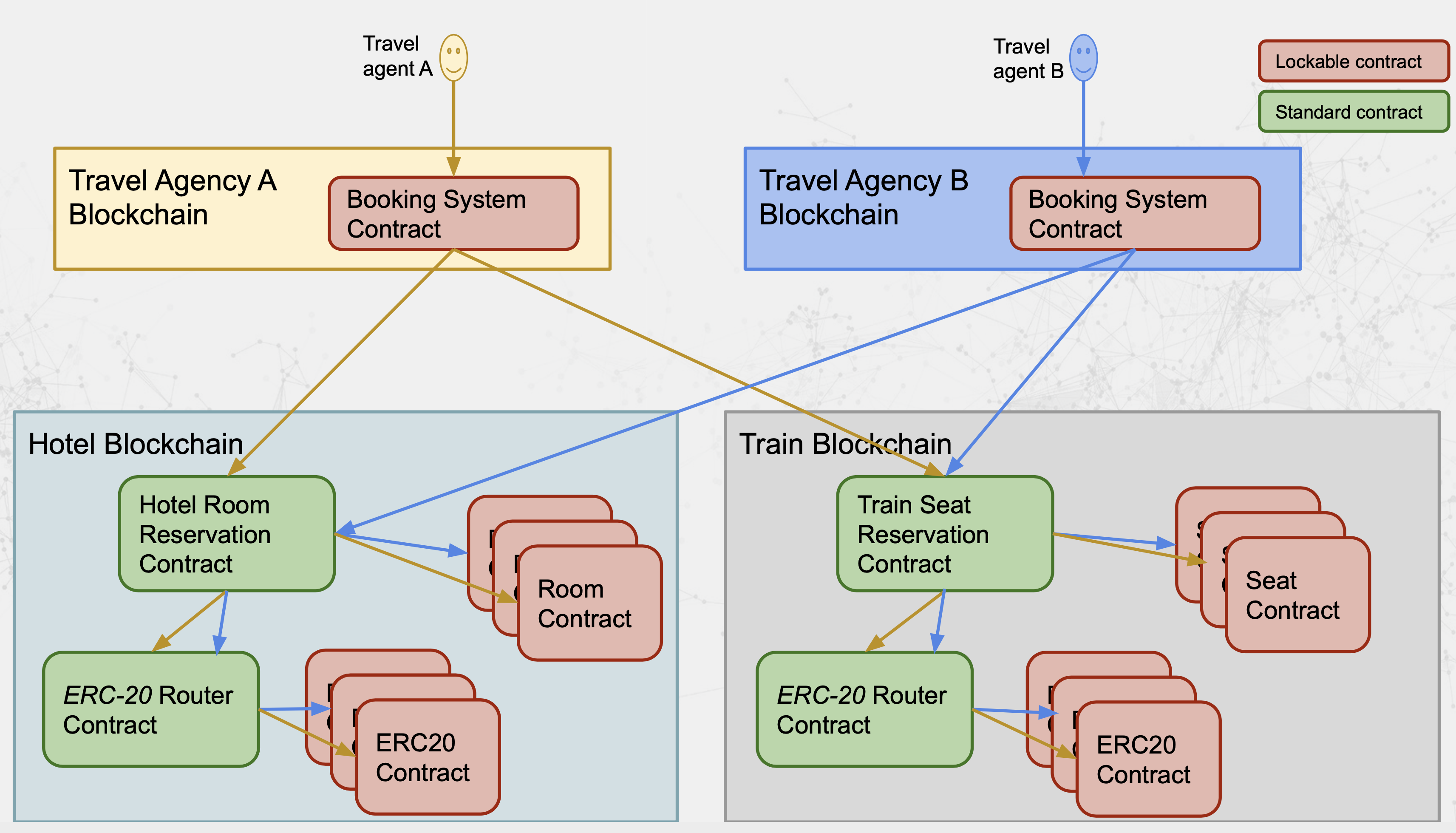}       
\caption{Travel Agent Use Case - Locking failure}
\label{fig:TravelAgentUseCase03}
\end{figure}

\subsection{Supply Chain Selective Privacy with Selective Transparency}

In many commercial applications, a vendor may wish to publish information to regulators to demonstrate regulatory compliance, or to customers to provide assurances of product provenance. With only a single, global, blockchain, all details of a transaction must be published and will be visible to all participants of the global blockchain. The vender may prefer, for example, to keep the identities of its suppliers secret from competing vendors. Selective privacy allows a vendor to publish only the subset of transaction information that is required by regulation or customer assurance, but maintain the secrecy of confidential business details.

An example scenario is illustrated in Fig.~\ref{fig:SupplyChainUseCase} on page~\pageref{fig:SupplyChainUseCase} where a customer is able to access details of the product being purchased from the vendor, such as its origin (place(s) of production and processing), and any certification (e.g. organic) which is noted on the product's label.

\begin{figure} [h]  
\centering
\includegraphics[width=\columnwidth]{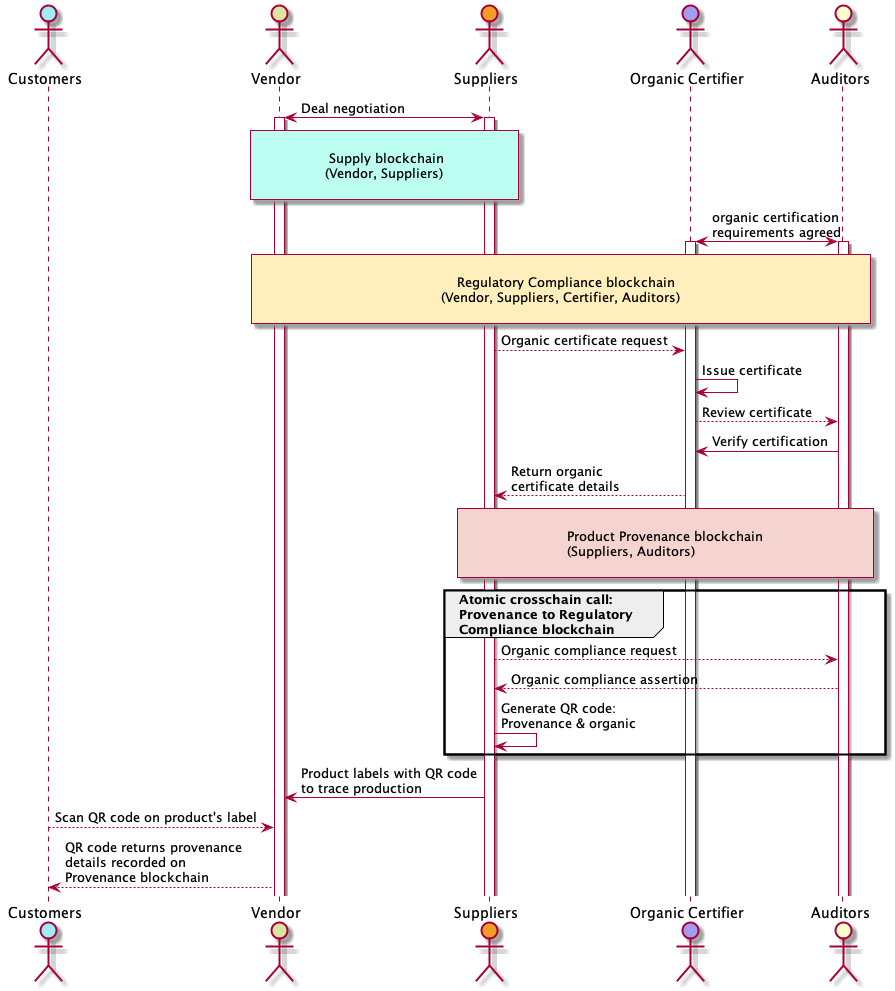}    
\caption{Supply Chain Use Case}
\label{fig:SupplyChainUseCase}
\end{figure}

In this scenario, there are three blockchains:
\begin{itemize}
\item A \textit{Supply Chain} blockchain which maintains all of the transactions between a vendor and its suppliers.
\item A \textit{Regulatory Compliance} blockchain that includes members of the Supply Chain, an Auditor and an Organic Certifier Regulator. This chain includes all data needed to guarantee regulatory compliance, but without the confidential business details found in the Supply Chain.
\item  A \textit{Product Provenance} blockchain which holds all the information required to assure customers of various aspects of the goods being purchased. For example, the authenticity of the product's claim of being organic or ``Fair Trade" goods, and the place of production of the product, including the entire (or partial) production chain of the goods. The latter is especially important with respect to a systems approach to pest risk management that encompasses more environmentally friendly techniques~\cite{Holt2018}.
\end{itemize}

The use of Atomic Crosschain Transactions in this scenario will enable the Product Provenance blockchain to request the compliance of a product which is purported to be organic. Only the status of the certification is returned to the Provenance chain from the Regulatory Compliance blockchain, with details of the assessment remaining private to members of the Regulatory Compliance blockchain. Based on the outcome of the atomic crosschain call, a QR code can be generated and added to the Product Provenance blockchain. The suppliers can use this QR code to print onto their product labels, which a customer can then scan to retrieve the product information.

Similarly, events in the Supply Chain blockchain can be published selectively to the both the Regulatory Compliance blockchain and the Product Provenance blockchain. This facilitates the application of both selective privacy and selective transparency. Moreover, Atomic Crosschain Transactions ensure that any piece of information that needs to appear in more than one blockchain appears in all the necessary blockchains or, if a failure occurs, in none of them.

\subsection{Reading from Oracle Chains}

An Oracle Chain is a blockchain that maintains a set of data that is valuable to other blockchain applications. One example is currency exchange when value is transferred between two blockchains that are using different currencies. The Oracle chain may serve as the authoritative source for the rate of exchange between the two currencies. Both the sender and receiver consult the Oracle Chain to be assured that the transaction is properly credited. Fig.~\ref{fig:OracleChainUseCase} illustrates the Oracle Chain use case.

\begin{figure} [h]
\centering
\includegraphics[width=0.7\columnwidth]{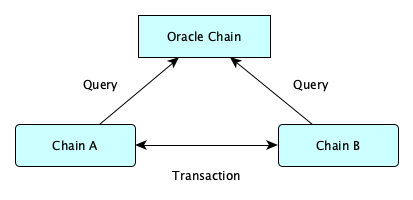}       
\caption{Oracle Chain Use Case}
\label{fig:OracleChainUseCase}
\end{figure}

A simple implementation of this scheme may lead to some risk if the exchange rate is volatile. Either by accident, or due to malicious interference, the two parties may obtain different values from the Oracle Chain. Without proper safeguards, this may cause an inconsistency in the recorded transaction. The possibility of such inconsistencies is unacceptable, and offers the opportunity for malicious actors to corrupt the integrity of the system.

If the consultation of the Oracle Chain is done as part of an Atomic Crosschain Transaction, we can ensure that both parties receive the same value from the Oracle chain, or in the case of an error, neither party receives a response. This guarantees that either the transaction completes, with a consistent value on both sides, or the transaction is aborted with no value transferred.

\subsection{Ephemeral Blockchain Archiving}
In practice blockchains may only be required to exist for a limited time, for example to execute a business deal or perform a particular function that is short-lived. To illustrate this functionality, we consider the sub-division of land. During the acquisition, development application, approval process, site works, site inspections and finally registering the blocks, various blockchains may exist. However, once plan sealing has taken place and all the blocks of land have been registered as separate titles, some blockchains are no longer relevant and become obsolete, while others may need to persist.

The example scenario in Fig.~\ref{fig:land-subdiv} on page~\pageref{fig:land-subdiv} demonstrates the use of Atomic Crosschain Transaction technology where the Developer on the \textit{Construction blockchain} submits a view call to the \textit{Development Application (DA) blockchain} to write relevant details of DA conditions affecting the work on site (e.g. tree clearing, tree retention, and sewer pipes) that need to be inspected by Council or the Utility companies (e.g. Electricity, Water, Sewer), to the DA chain. The project manager, also a participant of the \textit{Site Works blockchain}, is then able to check on the regulatory compliance of DA conditions and schedule inspections by the relevant Authorities. 

Moreover, progress of site construction work, early warning of potential delays, scheduling the different site works, and issuing progress claim assessments to the site operator to raise an invoice for progress payments, can all be recorded on the Construction blockchain. As inspections are completed, the outcomes are recorded on the Site works blockchain, which is periodically queried via a crosschain view call from the DA blockchain, so that the conditions that have been attached to the approval of the development application can be tracked and dates of inspections added to the DA blockchain. 

Once plan sealing has been obtained, the plan is submitted to the Land Registry office to process and issue titles for each of the newly created blocks. Although the plans have been sealed by Council, an atomic function call is necessary from the Land Registry blockchain to the DA blockchain to get compliance information, dates and other details before the Land Registry checks can be completed. All conditions need to be satisfied before titles can be issued. Once titles have been issued, they will be recorded on the Land Registry blockchain and the Site Works and DA blockchains can be archived.

The advantage of a blockchain approach to land sub-division with atomic crosschain capability is that  it provides irrefutable evidence of the compliance with development application conditions across different areas so that Land Registry can more readily process and register titles.

\begin{figure*}
  \includegraphics[width=0.8\linewidth]{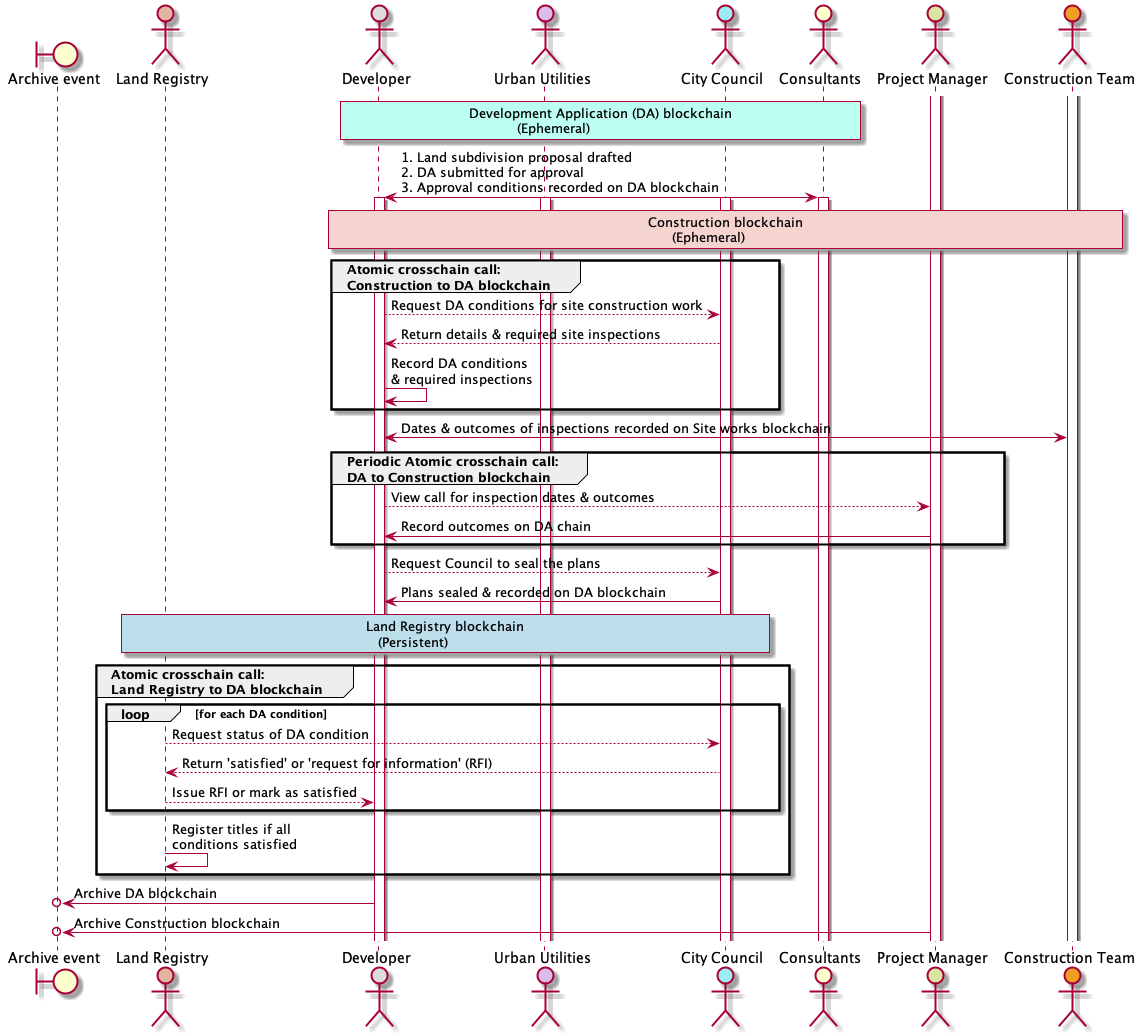}
  \caption{Blockchains for the sub-division of land}
  \label{fig:land-subdiv}
\end{figure*}

%% file: analysis.tex
\section{Analysis}
\label{sec:analysis}

Typically, any desired property of a distributed system can be expressed as a combination of safety and liveness properties. Similarly, the atomicity property of the atomic crosschain protocol can also be seen as a combination of safety and liveness properties. In general, the safety property captures the notion that bad states are not reachable, and the liveness property captures the notion that good states are eventually reached. For atomicity, the safety property translates to: when the protocol finishes, the whole crosschain transaction gets either committed or rolled back, the liveness property states that: the atomic crosschain transaction protocol eventually (after finite amount of time) terminates. We present a more formal discussion here.

{\bf Limit on the number of Byzantine nodes.}  Consensus protocols invariably has a limit on the number of Byzantine nodes in a network. For example, proof-of-work protocol requires the number of Byzantine nodes to be less than half of the total number of nodes. Similarly, IBFT and IBFT2 protocols have a limit of one-third the total number of nodes. Because our atomic crosschain protocol is built on top of consensus protocols, we require and {\bf assume} that the conditions for the correct operation of the consensus protocols are met. In other words, we inherit the same limit on the number of Byzantine nodes from the dependant consensus protocol.

\begin{theorem}[Safety]
Suppose that the atomic crosschain protocol has finished executing a crosschain transaction $c$. Then:
\begin{itemize}
\item if $c$ succeeds then the protocol successfully commits the state updates of all the associated (originating and subordinate) transactions,
\item if $c$ fails then the protocol rolls back the state updates of all the associated transactions.
\end{itemize}
\end{theorem}

\begin{proof}
We prove by case analysis of all possible states of the crosschain transaction $c$ recorded in the coordination contract.

\begin{itemize}
\item {\bf not-started} This is possible only if $c$ fails before the state of the coordination contract is updated to {\em started}. At this point neither the Originating Transaction nor the Subordinate Transactions has started to execute. Hence there is nothing to rollback. 

\item {\bf started} We show that this situation is not possible. The coordination contract has received the {\em start} message from the originating chain, but not the {\em commit} or {\em ignore} message. Based on the time-out specified in the {\em start} message, the nodes on the originating chain start their respective local timers, after processing the Originating Transaction. Similarly, the nodes on the subordinate chains start local timers after processing their respective Subordinate Transactions. After the timers expire, the nodes look up the state of $c$ in the coordination contract. In case of any failure in processing the Originating Transaction or in one of the Subordinate Transactions, the originating chain would have updated the coordination contract to {\em ignored} state. Also, the state of the coordination contract is updated to {\em ignored} state when there is a time-out. So, when $c$ is processed (either successfully or not) the state of $c$ in coordination contract can only end up in {\em committed} or {\em ignored}, after having been in the {\em started} state.

\item {\bf committed} It means that the coordination contract received the threshold signed {\em commit} message from the originating chain. After submitting the {\em commit} message, originating blockchain's coordinating node messages all the validators in its chain and other coordinating nodes on the Subordinate Transaction blockchains to check the state of the coordination contract. If it is Byzantine, then the local timers on the nodes (of all participating chains) expire and the nodes look up the coordination contract for the state of $c$, and find that $c$ is {\em committed}. They then commit their local provisional states. Hence the whole transaction $c$ is committed.

\item {\bf ignored} This is possible when either the coordinating node on the originating chain has sent an {\em ignore} message, or $c$ has timed out, and the coordination contract updates the state of $c$ to be {\em ignored}. Nodes will commit or roll back their provisional states only after looking up the coordination contract. In this case all nodes on all participating chains roll back their updates, implying that the whole crosschain transaction $c$ is rolled back. 
\end{itemize}
\end{proof}

\begin{theorem}[Liveness]
Suppose the atomic crosschain protocol is started, meaning, a crosschain transaction is submitted to a coordinating node on the originating blockchain. Then, the atomic crosschain protocol protocol will eventually (after a finite number of steps) finish.
\end{theorem}

\begin{proof}
The intuition of the proof is simple. The protocol itself has a finite number of steps. Also, as the {\em start} message is submitted, all validator nodes on all participating chains start their local timer. This ensures that the whole transaction will time out eventually, in case there are indefinite delays. 

We consider the special case when the crashing (or becoming Byzantine) of nodes might prevent the protocol from terminating.  The crashing of other nodes are governed by the dependant consensus protocol. Hence we assume that we always have a minimum number of nodes to be non-Byzantine as dictated by the consensus protocol. We examine if crashing a coordinating node introduces indefinite waits, which implies that the protocol does not terminate. 
\begin{itemize}
\item Suppose the coordinating node on the originating chain crashes before submitting the crosschain transaction {\em start} message to the coordination contract. Then no other chain knows about this crosschain transaction and this transaction is terminated and deemed to have failed. 

\item Suppose the coordinating node on the originating chain crashes after submitting the {\em start} message but before submitting the {\em commit} message, then the whole transaction times out. The local timers on validator nodes of the originating blockchain time out, look up the coordination contract, find the transaction is {\em ignored}, and abort the transaction. The subordinate chains complete the execution of views or transactions and send their Subordinate Transaction {\em ready} messages to the coordinating node of the originating chain. Because the coordinating node of originating chain has crashed, they observe that the state of the coordinating contract is never changed to {\em committed}, and the validator nodes on all participating chains will eventually time out, look up the coordination contract (which is set to {\em ignored} after time out), and abort the transaction.
 
\item Suppose the coordinating node on originating chain crashes after submitting the crosschain transaction {\em commit} message to the coordination contract. Then (validator nodes on) all chains see the {\em committed} state of coordinating contract and will commit the transaction. 

\item Suppose the coordinating node on a subordinate chain crashes. Then it cannot send the Subordinate Transaction {\em ready} message to the originating chain. The coordinating node on the originating chain keeps waiting and times out. The coordination contract gets updated to {\em ignored} state and the transaction terminates and is rolled back.
\end{itemize}
\end{proof}

Establishing the liveness property is very useful as it implies the properties of termination and deadlock-free.

%% file: conclusion.tex
\section{Conclusion}
Atomic Crosschain Transactions allow application developers to create complex cross-blockchain applications in a straightforward manner. It does this by absorbing the vast majority of the complexity required to deliver this behaviour using a novel combination of technologies.

%% file: acknowledgements.tex
\ifCLASSOPTIONcompsoc
  \section*{Acknowledgments}
\else
  \section*{Acknowledgment}
\fi
This research has been undertaken whilst we have been employed full-time at ConsenSys. Peter acknowledges the support of University of Queensland where he is completing his PhD, and in particular the support of my PhD supervisor Dr Marius Portmann.

We acknowledge the co-authors of the original Atomic Crosschain Transaction paper \cite{robinson2019b} Dr David Hyland-Wood and Roberto Saltini for their help creating the technology upon which this paper is based. We thank Dr Catherine Jones and Horacio Mijail Anton Quiles for reviewing this paper and providing astute feedback.

%% file: AtomicCrosschainTransactions_01.bbl
\begin{thebibliography}{10}
\providecommand{\url}[1]{#1}
\csname url@samestyle\endcsname
\providecommand{\newblock}{\relax}
\providecommand{\bibinfo}[2]{#2}
\providecommand{\BIBentrySTDinterwordspacing}{\spaceskip=0pt\relax}
\providecommand{\BIBentryALTinterwordstretchfactor}{4}
\providecommand{\BIBentryALTinterwordspacing}{\spaceskip=\fontdimen2\font plus
\BIBentryALTinterwordstretchfactor\fontdimen3\font minus
  \fontdimen4\font\relax}
\providecommand{\BIBforeignlanguage}[2]{{%
\expandafter\ifx\csname l@#1\endcsname\relax
\typeout{** WARNING: IEEEtran.bst: No hyphenation pattern has been}%
\typeout{** loaded for the language `#1'. Using the pattern for}%
\typeout{** the default language instead.}%
\else
\language=\csname l@#1\endcsname
\fi
#2}}
\providecommand{\BIBdecl}{\relax}
\BIBdecl

\bibitem{robinson2019b}
\BIBentryALTinterwordspacing
P.~Robinson, D.~Hyland-Wood, R.~Saltini, S.~Johnson, and J.~Brainard, ``{Atomic
  Crosschain Transactions for Ethereum Private Sidechains},'' 2019. [Online].
  Available: \url{https://arxiv.org/abs/1904.12079}
\BIBentrySTDinterwordspacing

\bibitem{robinson2018a}
\BIBentryALTinterwordspacing
P.~Robinson, ``{Requirements for Ethereum Private Sidechains},'' 2018.
  [Online]. Available: \url{http://adsabs.harvard.edu/abs/2018arXiv180609834R}
\BIBentrySTDinterwordspacing

\bibitem{hashtimelock}
\BIBentryALTinterwordspacing
``{Hashed Timelock Contracts},'' 2017. [Online]. Available:
  \url{https://en.bitcoin.it/wiki/Hashed{\_}Timelock{\_}Contracts}
\BIBentrySTDinterwordspacing

\bibitem{cosmos2016}
\BIBentryALTinterwordspacing
E.~Buchman and J.~Kwon, ``{Cosmos: A network of distributed ledgers.}'' Github,
  2016. [Online]. Available:
  \url{https://github.com/cosmos/cosmos/blob/master/WHITEPAPER.md}
\BIBentrySTDinterwordspacing

\bibitem{Clearmatics2018c}
\BIBentryALTinterwordspacing
Clearmatics, ``{General interoperability framework for trustless cross-system
  interaction},'' 2018. [Online]. Available:
  \url{https://github.com/clearmatics/ion}
\BIBentrySTDinterwordspacing

\bibitem{btc-relay}
\BIBentryALTinterwordspacing
J.~Chow, ``{BTC Relay},'' 2016. [Online]. Available:
  \url{https://media.readthedocs.org/pdf/btc-relay/latest/btc-relay.pdf}
\BIBentrySTDinterwordspacing

\bibitem{crosschain-github}
\BIBentryALTinterwordspacing
PegaSys, ``{Atomic Crosschain Transaction Sample Code Github Repository}.''
  [Online]. Available: \url{https://github.com/PegaSysEng/sidechains-samples}
\BIBentrySTDinterwordspacing

\bibitem{wood2016a}
\BIBentryALTinterwordspacing
G.~Wood, ``{Ethereum: a secure decentralized generalised transaction ledger},''
  Github, p. Github site to create pdf, 2016. [Online]. Available:
  \url{https://github.com/ethereum/yellowpaper}
\BIBentrySTDinterwordspacing

\bibitem{bls-threshold}
\BIBentryALTinterwordspacing
A.~Boldyreva, ``{Threshold Signatures, Multisignatures and Blind Signatures
  Based on the Gap-Diffie-Hellman-Group Signature Scheme}.'' [Online].
  Available: \url{http://www-cse.ucsd.edu/users/aboldyre}
\BIBentrySTDinterwordspacing

\bibitem{bls-threshold-youtube}
\BIBentryALTinterwordspacing
P.~Robinson, ``{Ethereum Engineering Group: BLS Threshold Signatures -
  YouTube}.'' [Online]. Available:
  \url{https://www.youtube.com/watch?v=XZTvBYG9pn4}
\BIBentrySTDinterwordspacing

\bibitem{shamir1979}
\BIBentryALTinterwordspacing
A.~Shamir, ``{How to share a secret},'' \emph{Commun. ACM}, vol.~22, no.~11,
  pp. 612--613, 1979. [Online]. Available:
  \url{https://cs.jhu.edu/{~}sdoshi/crypto/papers/shamirturing.pdf}
\BIBentrySTDinterwordspacing

\bibitem{bls2004}
\BIBentryALTinterwordspacing
D.~Boneh, B.~Lynn, and H.~Shacham, ``{Short Signatures from the Weil
  Pairing},'' \emph{Journal of Cryptology}, vol.~17, no.~4, pp. 297--319, 2004.
  [Online]. Available: \url{https://doi.org/10.1007/s00145-004-0314-9}
\BIBentrySTDinterwordspacing

\bibitem{ped1991}
\BIBentryALTinterwordspacing
T.~P. Pedersen and D.~W. Davies, ``{A Threshold Cryptosystem without a Trusted
  Party},'' in \emph{Advances in Cryptology — EUROCRYPT '91}.\hskip 1em plus
  0.5em minus 0.4em\relax Berlin, Heidelberg: Springer Berlin Heidelberg, 1991,
  pp. 522--526. [Online]. Available:
  \url{https://link.springer.com/chapter/10.1007/3-540-46416-6{\_}47}
\BIBentrySTDinterwordspacing

\bibitem{Holt2018}
J.~Holt, A.~Leach, S.~Johnson, D.~Tu, D.~Nhu, N.~Anh, M.~Quinlan, P.~Whittle,
  K.~Mengersen, and J.~Mumford, ``{Bayesian Networks to Compare Pest Control
  Interventions on Commodities Along Agricultural Production Chains},''
  \emph{Risk Analysis}, vol.~38, no.~2, 2018.

\end{thebibliography}
